\documentclass[11pt,letterpaper]{article}
\usepackage[utf8]{inputenc}
\usepackage{calc}
\usepackage{amsmath,amssymb,amsthm}
\usepackage{anysize}
\marginsize{1in}{1in}{1.6cm}{2cm}
\usepackage{tikz}
\usetikzlibrary{fadings}
\usetikzlibrary{decorations.pathreplacing}
\usepackage{enumerate}
\usepackage[colorinlistoftodos]{todonotes}
\usepackage{hyperref}

\newcommand{\importfading}[2]{%
#1
}

\def\eg{{e.\,g.}}
\def\ie{{i.\,e.}}
\def\BR{\mathbb{R}}
\def\eps{\varepsilon}

\def\dmax{D}
\def\dmin{d_{\min}}
\def\fmax{b}
\def\fmin{a}
\def\yopt{y_{\text{opt}}}
\def\th{\text{th}}
\def\RLP{\textsc{Ring Loading Problem}}
\def\Partition{\textsc{Partition}}
\def\NP{\text{NP}}

\newtheorem{theorem}{Theorem}
\newtheorem{lemma}{Lemma}
\newtheorem{observation}{Observation}

\newtheorem{definition}{Definition}
\newtheorem{conjecture}{Conjecture}

\title{A note on the ring loading problem\thanks{Large parts of the research described in this paper were conducted while the author was visiting the University of British Columbia (UBC) in Vancouver, Canada. It was partially supported by the DFG Research Center \textsc{Matheon} ``Mathematics for key technologies'' in Berlin, by DFG Priority Programme 1307 ``Algorithm Engineering'' (grant SK~58/6-3), by the Alexander von Humboldt-Foundation, and by the Peter Wall Institute for Advanced Studies at UBC.}}
\author{Martin Skutella\footnote{TU Berlin, Institute of Mathematics, MA 5-2, Stra{\ss}e des 17.~Juni 136, 10623 Berlin, Germany, \href{mailto:martin.skutella@tu-berlin.de}{\texttt{martin.skutella@tu-berlin.de}}}}
\date{}

\colorlet{myred}{red!100!black}
\colorlet{myredlabel}{myred!100!black}
\colorlet{myblue}{blue!100!black}
\colorlet{mybluelabel}{myblue!100!black}
\colorlet{mygreen}{green!80!black}
\colorlet{mygreenlabel}{mygreen!80!black}

\tikzstyle{point}=[circle, fill, inner sep = 0pt, minimum size = 0.4em]

\tikzstyle{node}=[circle, fill, inner sep = 0pt, minimum size = 0.9em]
\tikzstyle{bignode}=[circle, fill, inner sep = 0pt, minimum size = 1.3em,text=white]

\newcommand{\drawpattern}[6]{

\pgfmathparse{0}\let\m\pgfmathresult
\pgfmathparse{#2}\let\y\pgfmathresult
\pgfmathparse{#2}\let\ma\pgfmathresult
\pgfmathparse{#2}\let\mi\pgfmathresult

\foreach \z in {#1} {
\pgfmathparse{int(\m+1)}\global\let\m\pgfmathresult
\pgfmathparse{\y+\z}\global\let\y\pgfmathresult
\pgfmathparse{max(\y,\ma)}\global\let\ma\pgfmathresult
\pgfmathparse{min(\y,\mi)}\global\let\mi\pgfmathresult
}

\fill[gray!40!white] (0,\mi) rectangle (\m,\ma); 
\draw[very thick,->] (-0.2,0) -- (\m+0.3,0) node [right] {$k$}; 
\foreach \i in {0,...,\m} \draw [thick] (\i,-0.1) node[label=below:$\i$] {} -- +(0,0.2); 
\draw[very thick,->] (0,0) -- (0,\ma+0.3) node [above] {$p_z(k)$}; 
\draw[very thick] (0,#2) node[point,label=left:$#3$] {} \foreach \z in {#1} {-- ++(1,\z) node[point] {}}; 
\draw[thick,dashed] (0,\y) node[label=left:$#4$] {} -- +(\m,0); 
\draw[thick] (-0.1,\mi) node[label=left:$#5$] {} -- +(0.2,0);
\draw[thick] (-0.1,\ma) node[label=left:$#6$] {} -- +(0.2,0);
}

\def\radius{3.51}
\pgfmathparse{\radius+0.45}\let\blacklabelradius\pgfmathresult
\pgfmathparse{\radius+0.35}\let\graylabelradius\pgfmathresult
\pgfmathparse{\radius+0.07}\let\grayarcradius\pgfmathresult
\pgfmathparse{\radius-0.3}\let\flowradius\pgfmathresult
\def\grayangle{12}

\newcommand{\drawring}[2]{
\begin{scope}[xscale=-1]
\pgfmathparse{0}\let\m\pgfmathresult
\pgfmathparse{0}\let\flow\pgfmathresult
\pgfmathparse{0}\let\sumdemands\pgfmathresult
\foreach \u/\v in {#1} {
\pgfmathparse{int(\m+1)}\global\let\m\pgfmathresult
\pgfmathparse{int(\flow+\u)}\global\let\flow\pgfmathresult
\pgfmathparse{int(\sumdemands+\u+\v)}\global\let\sumdemands\pgfmathresult
}
\pgfmathparse{180.0/\m}\let\delta\pgfmathresult
\pgfmathparse{40.0/\m}\global\let\graylabelangle\pgfmathresult

\draw [ultra thick] (0,0) circle (\radius cm);

\pgfmathparse{-1}\let\i\pgfmathresult
\foreach \u/\v in {#1}{
\pgfmathparse{int(\i+1)}\global\let\i\pgfmathresult
\pgfmathparse{(\i+0.5)*\delta}\let\angle\pgfmathresult

\draw [very thick,myred,->] (\angle:\grayarcradius cm) arc (\angle:\angle+\grayangle:3.07cm);
\node [myredlabel] at (\angle+\graylabelangle:\graylabelradius cm) {$\v$};
\draw [very thick,mygreen,->] (\angle:\grayarcradius cm) arc (\angle:\angle-\grayangle:3.07cm);
\node [mygreenlabel] at (\angle-\graylabelangle:\graylabelradius cm) {$\u$};
\pgfmathparse{int(\i+1)}\let\im\pgfmathresult
\ifthenelse{\equal{#2}{}}{
\node [node] (node\i) at (\angle:\radius) {};
}{
\node [bignode] (node\i) at (\angle:\radius) {\footnotesize$\im$};
}
\pgfmathparse{int(\i+\m+1)}\let\im\pgfmathresult
\ifthenelse{\equal{#2}{}}{
\node [node] (opposite_node\i) at (180+\angle:\radius) {};
}{
\node [bignode] (opposite_node\i) at (180+\angle:\radius) {\footnotesize$\im$};
}
\draw [thick,gray,dash pattern=on 4pt off 3.92pt] (node\i) -- (opposite_node\i);

\pgfmathparse{int(\u+\v)}\global\let\y\pgfmathresult
\node at (180+\angle:\blacklabelradius cm) {$\y$};

\pgfmathparse{int(\flow-\u+\v)}\global\let\flow\pgfmathresult
\node [mybluelabel] at (\angle+0.5*\delta:\flowradius cm) {$\flow$};
\pgfmathparse{int(\sumdemands-\flow)}\let\g\pgfmathresult
\node [mybluelabel] at (180+\angle+0.5*\delta:\flowradius cm) {$\g$};
}
\end{scope}
}

\begin{document}
\maketitle

\begin{abstract}
The \RLP{} is an optimal routing problem arising in the planning of optical communication networks which use bidirectional SONET rings. In mathematical terms, it is an unsplittable multicommodity flow problem on undirected ring networks. We prove that any split routing solution to the \RLP{} can be turned into an unsplittable solution while increasing the load on any edge of the ring by no more than~$+\frac{19}{14}\dmax$, where~$\dmax$ is the maximum demand value. This improves upon a classical result of Schrijver, Seymour, and Winkler (1998) who obtained a slightly larger bound of~$+\frac32\dmax$. We also present an improved lower bound of~$\frac{11}{10}\dmax$ (previously~$\frac{101}{100}\dmax$) on the best possible bound and disprove a famous long-standing conjecture of Schrijver et al.\ in this context.
\end{abstract}

\section{Introduction}

An instance of the \RLP{} is given by an undirected ring (cycle) on node set~$V=\{1,2,\dots,n\}$ (numbered clockwise along the ring) with demands~$d_{i,j}\geq0$ for each pair of nodes~$i<j$ in~$V$. The task is to route all demands unsplittably, that is, each demand~$d_{i,j}$ needs to be routed from node~$i$ to node~$j$ either in clockwise direction on the path~$i,i+1,\dots,j$ or in counterclockwise direction on the path~$i,i-1,\dots,j$. The objective is to minimize the maximum load on an edge of the ring. The optimum solution value is denoted by~$L$. A small example instance of the \RLP{} is given in Figure~\ref{fig:tiny-example}.
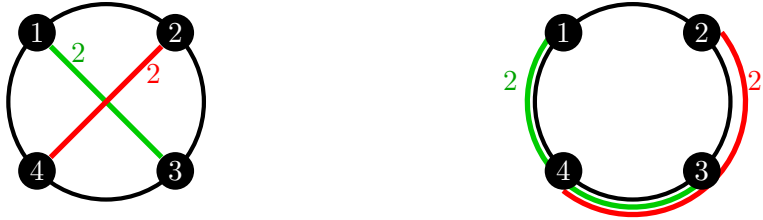
\begin{figure}[t]
\centering
\begin{tikzpicture}
\def\radius{1.3}
\begin{scope}[xscale=-1]
\draw [ultra thick] (0,0) circle (\radius cm);	
\foreach \i in {1,...,4} \node [bignode] (node\i) at (-45+90*\i:\radius cm) {$\i$};
\draw [line width=2pt,mygreen] (node1) -- node[near start,above,mygreenlabel] {$2$} (node3);
\draw [line width=2pt,myred] (node2) -- node[near start,right,myredlabel] {$2$} (node4);
\end{scope}
\begin{scope}[xshift=7cm,xscale=-1]
\draw [line width=2pt,mygreen] (-7+45:1.4cm) arc (-7+45:7-135:1.4cm);
\node [mygreenlabel] at (10:1.65cm) {$2$};
\draw [line width=2pt,myred] (7+135:1.5cm) arc (7+135:-7+315:1.5cm);
\node [myredlabel] at (170:1.65cm) {$2$};
\draw [ultra thick] (0,0) circle (\radius cm);	
\foreach \i in {1,...,4} \node [bignode] (node\i) at (-45+90*\i:\radius cm) {$\i$};
\end{scope}
\end{tikzpicture}
\caption{Left: An instance of the \RLP{} on a ring of $n=4$ nodes and two nonzero demands~$d_{1,3}=d_{2,4}=2$. Right: An unsplittable routing where the first demand is routed counterclockwise from node~$1$ to node~$3$, and the second demand is routed clockwise. The maximum edge load is~$4$ and this is obviously optimal, \ie, $L=4$.}
\label{fig:tiny-example}
\end{figure}

\paragraph{Application in telecommunication networks.}
This optimization problem arose in the early 1990ies as a crucial subproblem in the design of survivable telecommunication networks utilizing fiber-optic-based technologies; it was first studied by Cosares and Saniee~\cite{CosaresSaniee1994} who also introduced the name \emph{\RLP}. 

Fiber-optic-based telecommunication networks are built according to a standard called SONET (Synchronous Optical Networks) or the later derived SDH (Synchronous Digital Hierarchy) standard. These networks are based on rings of nodes (switches) with optical links between each pair of neighboring nodes. The actual bandwidth available along any edge of a ring is essentially determined by the capacity of the add/drop multiplexers (ADMs) installed at its nodes which make up a large part of the hardware cost. An important goal in network planning is thus to minimize the bandwidth required to satisfy all demands. In the early 1990ies Bell Communications Research (Bellcore) developed a network planning software called the `SONET Toolkit' which also initiated the scientific study of the \RLP{}. For a more detailed account of the technical background and history, we refer to the articles~\cite{CosaresSaniee1994,SchrijverSeymourWinkler-SIAMReview,Khanna-SONET97,Shepherd-survey2009}.

\paragraph{Known results.}
Cosares and Saniee~\cite{CosaresSaniee1994} prove by a reduction from the problem \Partition{} that the \RLP{} is weakly \NP-hard. It is not known whether the problem is even strongly \NP-hard or can be solved in pseudo-polynomial time. If all non-zero demands are equal, the problem can be solved in polynomial time; this follows from the work of Frank~\cite{Frank1985} (see also~\cite{FrankNishizekiSaitoSuzukiTardos1992}) and is based on a theorem of Okamura and Seymour~\cite{OkamuraSeymour81}. 

The hardness of the general problem motivates the consideration of the relaxed version of the \RLP{} where demands may be split, \ie, a demand can be sent partly clockwise, partly counterclockwise. The optimum ring load~$L^*$ of a split routing is obviously a lower bound on the optimum load~$L$ of an unsplittable routing. Cosares and Saniee~\cite{CosaresSaniee1994} prove that~$L\leq2L^*$, and this inequality is tight; see, \eg, the instance depicted in Figure~\ref{fig:tiny-example} where~$L=4$ and~$L^*=2$. Moreover, they present an efficient algorithm for obtaining an unsplittable routing of load at most~$2L^*$. For the problem of finding an optimum split routing, Myung, Kim, and Tcha~\cite{MyungKimTcha1997} give a clever combinatorial algorithm with running time~$O(nk)$, where~$k$ is the number of nonzero demands. 

Schrijver, Seymour, and Winkler, in a landmark paper~\cite{SchrijverSeymourWinkler-SIDMA} (see also the annotated reprint~\cite{SchrijverSeymourWinkler-SIAMReview}), present a clever analysis for a simple greedy algorithm that turns any split routing into an unsplittable routing while increasing the load on any edge by at most~$\frac32\dmax$, where~$\dmax:=\max_{i,j}d_{i,j}$ is the maximum demand value. Their result thus implies~$L\leq L^*+\frac32\dmax$. Khanna~\cite{Khanna-SONET97} shows how to use this bound in order to obtain a polynomial-time approximation scheme for the \RLP{}. For the special case of uniform nonzero demands~$d_{i,j}=\dmax$, Schrijver et al.\ even prove~$L\leq L^*+\dmax$. 

On the other hand, Schrijver et al.\ exhibit an instance of the \RLP{} together with a carefully chosen split routing 
that cannot be turned into an unsplittable routing without increasing the load on some edge by at least~$\frac{101}{100}\dmax$. This observation, however, does \emph{not} imply a gap strictly larger than~$\dmax$ between the optimum values of split and unsplittable routings. In the conclusion of their paper, Schrijver, Seymour, and Winkler write:
\begin{quote}
``\dots even though the mathematics refuses to cooperate, \emph{we} guarantee~$L\leq L^*+D$.''	
\end{quote}
In his excellent survey~\cite{Shepherd-survey2009}, Shepherd restates this `guarantee' as Conjecture~1 and indicates its `tempting similarity' to a result of Dinitz, Garg, and Goemans~\cite{DGG99} on single-source unsplittable flows and a related conjecture of Goemans~\cite{BellairsOpenProblems} (see also~\cite{Skut-MathProg2002,MartensSalazarSkut06}).

\paragraph{Our contribution.} 
Our main result is the following theorem.

\begin{theorem}\label{thm:7/5}
Any split routing solution to the \RLP{} can be turned into an unsplittable routing while increasing the load on any edge by no more than~$\frac{19}{14}\dmax$. 
\end{theorem}

In particular, this result implies~$L\leq L^*+\frac{19}{14}\dmax$. Our algorithm runs in linear time and combines pairs of solutions obtained by the greedy algorithm of Schrijver, Seymour, and Winkler~\cite{SchrijverSeymourWinkler-SIDMA} using a clean crossover operation. 
In the area of genetic algorithms, it is known that the additional use of crossover (also called `recombination' or `sexual reproduction') is empirically more effective than mutation only, and this can sometimes even be confirmed by theoretical evidence; see, \eg,~\cite{JansenWegener2002,Theile-diss}. We are, however, not aware of previous approximation results relying on such crossover operations.

We also exhibit a relatively simple instance of the \RLP{} together with particular split routings that cannot be turned into unsplittable routings without increasing the load on some edge by at least~$\frac{11}{10}\dmax$
. After more than 15 years, our results are the first improvements on the classical results of Schrijver, Seymour, and Winkler~\cite{SchrijverSeymourWinkler-SIDMA} (upper bound $\frac32\dmax$ and lower bound $\frac{101}{100}\dmax$). Last but not least, we present an instance with~$L=L^*+\frac{11}{10}\dmax$
, thus disproving Schrijver et al.'s long-standing conjecture~$L\leq L^*+\dmax$.

We finally mention an interesting combinatorial implication of our results. Schrijver et al.~\cite{SchrijverSeymourWinkler-SIDMA} define~$\beta$ to be the infimum of all~$\alpha\in\BR$ such that the following combinatorial statement holds: for all positive integers~$m$ and nonnegative reals~$u_1,\dots,u_m$ and~$v_1,\dots,v_m$ with~$u_i+v_i\leq1$, there exist~$z_1,\dots,z_m$ such that for every~$k$, $z_k\in\{v_k,-u_k\}$ and
\begin{align*}
\left\lvert\sum_{i=1}^kz_i-\sum_{i=k+1}^mz_i\right\rvert\leq\alpha\enspace.
\end{align*} 
Schrijver et al.\ prove that~$\beta\in[\frac{101}{100},\frac32]$. As a consequence of our results, we narrow this interval down to~$\beta\in[\frac{11}{10},\frac{19}{14}]$.

\paragraph{Outline.}
In Section~\ref{sec:preliminaries} we present preliminary observations and results. In Section~\ref{sec:medium} we show how to improve the bound of~$\frac32\dmax$ in case there happens to be a demand of `medium size'. Subsequently, in Section~\ref{sec:small-big}, we deal with the remaining instances and complete the proof of Theorem~\ref{thm:7/5}. Finally, in Section~\ref{sec:lower-bound} we present the improved lower bound as well as the counterexample to Schrijver et al.'s conjecture.


\section{Preliminaries}
\label{sec:preliminaries}

In this section we present preliminary observations which can already be found --~explicitly or implicitly~-- in the work of Schrijver et al.~\cite{SchrijverSeymourWinkler-SIDMA}.
Two demands~$d_{i,j}$ and~$d_{k,\ell}$ are said to be \emph{parallel} if there is an $i$-$j$-path and a $k$-$\ell$-path that are edge-disjoint; otherwise they are said to \emph{cross}.

\pagebreak
\begin{observation}
\label{obs:parallel}
For an arbitrary split routing of two parallel demands, flow can be rerouted such that one of the two demands is routed unsplittably without increasing the load on any edge. 
\end{observation}

\begin{proof}
For~$q=1,2$, let~$P^a_{q},P^b_{q}$ be the two paths of the $q^{\th}$ demand, and let~$x^a_q$ and~$x^b_q$ be the flow routed along these paths, respectively. Since the two demands are parallel, we may assume that~$P^a_1$ and~$P^a_2$ are edge-disjoint. Notice that~$P^b_1$ thus contains~$P^a_2$, and~$P^b_2$ contains~$P^a_1$. Increasing flow on~$P^a_1$ and~$P^a_2$ by~$\min\{x^b_1,x^b_2\}$ and decreasing flow on~$P^b_1$ and~$P^b_2$ by the same amount therefore yields the desired routing.
\end{proof}

In order to prove Theorem~\ref{thm:7/5}, demands that are unsplittably routed in the given solution can be ignored (deleted). Thus, as a consequence of Observation~\ref{obs:parallel}, it suffices to prove the theorem for instances without parallel demands. In particular, any node is endpoint of at most one demand. Moreover, if some node of the ring is not an endpoint of any demand, its two incident edges always have the same load such that they can be merged and their common node be deleted. 

After this preprocessing, we are left with a ring of even length~$n=2m$ and~$m$ demands~$d_i:=d_{i,i+m}>0$, $i=1,2,\dots,m$. Moreover, a split routing is given where, for each $i=1,\dots,m$, exactly~$u_i>0$ units of flow are sent from~$i$ to~$i+m$ in clockwise direction and the remaining $v_i:=d_i-u_i>0$ units are sent counterclockwise; an example is given in Figure~\ref{fig:first-example}.
\begin{figure}[t]
\centering
\def\radius{3.53}
\begin{tikzpicture}
\drawring{3/3, 4/6, 4/4, 6/4, 3/3, 6/4, 3/1, 6/4}{1}
\end{tikzpicture}
\caption{An instance of the \RLP{} on $n=16$ nodes with $m=8$ nonzero demands for opposite node pairs and a split routing. There is, for example, a demand of value~$d_7=d_{7,15}=4$ of which~$u_7=1$ unit is routed clockwise and~$v_7=3$ units are routed counterclockwise from node~$7$ to node~$15$. The maximum demand value is~$\dmax=10$ and the maximum edge load is~$37$ (on edges~$\{2,3\}$ and~$\{3,4\}$).}
\label{fig:first-example}
\end{figure}
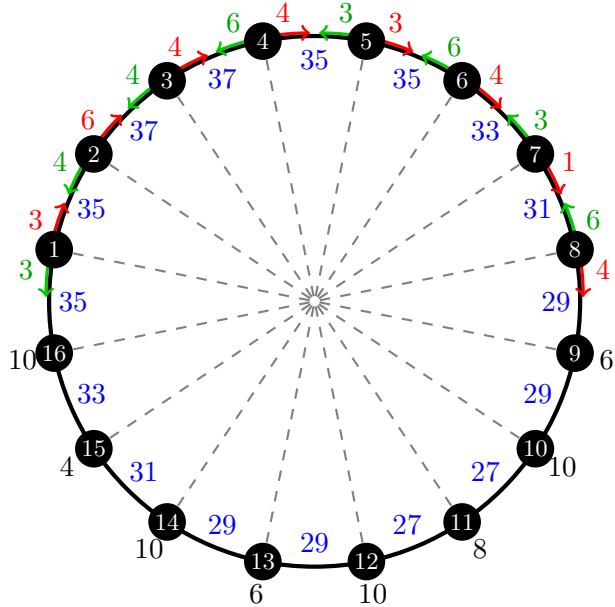

An unsplittable solution to the \RLP{} can be represented by an~$m$-tuple $z=(z_1,z_2,\dots,z_m)$ with $z_i\in\{v_i,-u_i\}$ for~$i=1,\dots,m$. Here, $z_i=v_i$ means that demand~$i$ is routed clockwise, \ie, compared to the given split routing, the flow on this path is increased by~$v_i$ while the flow on the counterclockwise path is decreased by~$v_i$. Correspondingly, $z_i=-u_i$ means that demand~$i$ is routed counterclockwise such that the flow on the clockwise path is decreased by~$u_i$ while the flow on the counterclockwise path is increased by~$u_i$. As a consequence, for~$k=1,\dots,m$, the load on edge~$\{k,k+1\}$ changes by 
\begin{align*}
\sum_{i=1}^kz_i-\sum_{i=k+1}^mz_i\enspace,
\end{align*}
and the load on the opposite edge~$\{k+m,1+(k+m\mod 2m)\}$ changes by the negative of this amount. Thus, compared to the given split routing, the maximum increase of flow on any edge is
\begin{align*}
\max_{k=1,\dots,m}\,\left\lvert\,\sum_{i=1}^kz_i-\sum_{i=k+1}^mz_i\,\right\rvert\enspace.
\end{align*}
We refer to this value as the \emph{additive performance} of solution~$z$.

For the sake of a more intuitive understanding, we use the following alternative representation of an unsplittable solution~$z$: For an arbitrary but fixed~$x\in\BR$, let 
\begin{align*}
p_z(k):=x+\sum_{i=1}^kz_i\qquad \text{for $k=0,1,\dots,m$.}
\end{align*}
An intuitive graphical representation of~$p_z$ is the graph of the continuous function on~$[0,m]$ obtained from~$p_z$ by linear interpolation; see Figure~\ref{fig:pattern} for an illustrating example.
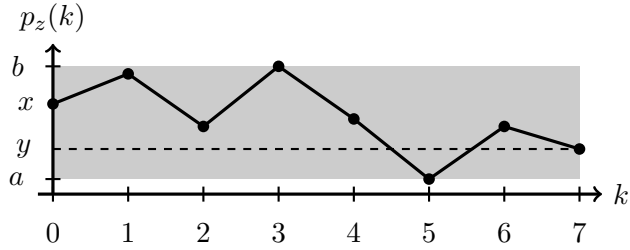
\begin{figure}[t]
\centering
\newcommand{\firstexample}{0.4,-0.7,0.8,-0.7,-0.8,0.7,-0.3}
\begin{tikzpicture}
\drawpattern{0.4,-0.7,0.8,-0.7,-0.8,0.7,-0.3}{1.2}{x}{y}{a}{b}
\end{tikzpicture}
\caption{Pattern~$p_z$ living on strip~$[a,b]$ with start point~$x$, end point~$y$.
}
\label{fig:pattern}
\end{figure}
Notice that
\begin{align*}
p_z(k)-p_z(k-1)=z_k\in\{v_k,-u_k\}\qquad\text{for~$k=1,\dots,m$.}	
\end{align*}
We refer to~$p_z$ as a \emph{pattern starting at $x=p_z(0)$ and ending at $y=p_z(m)$}. Let~$a:=\min_{k=0,\dots,m}p_z(k)$ and~$b:=\max_{k=0,\dots,m}p_z(k)$. In reference to Figure~\ref{fig:pattern} we refer to the interval~$[a,b]$ as a \emph{strip} and say that the \emph{pattern lives on strip~$[a,b]$ of width~$b-a$}. 

In the remainder of this paper we use the terms `pattern' and `unsplittable solution' interchangeably. Notice that an unsplittable solution corresponds to several patterns as the start point~$x$ can be arbitrarily chosen. All these patterns are, however, identical up to vertical shifting. To make the correspondence between unsplittable solutions~$z$ and patterns~$p_z$ unique, one could, for example, require~$x=0$ or~$a=0$. As these requirements would cause unnecessary technicalities in the presentation, we rather live with the ambiguity which doesn't cause any further problems.

\pagebreak
\begin{observation}\label{obs:performance}
Let~$z$ be an unsplittable solution and consider a corresponding pattern~$p_z$ living on strip~$[a,b]$ with start point~$x$ and end point~$y$.
\begin{enumerate}[(i)]
\item\label{obs:performance:i} The unsplittable solution~$z$ and its pattern~$p_z$ have additive performance
\begin{align}\label{eq:pattern-performance}
\max_{k=1,\dots,m}\left\lvert\sum_{i=1}^kz_i-\sum_{i=k+1}^mz_i\right\rvert = \max\bigl\{2\fmax-x-y,x+y-2\fmin\bigr\}\enspace.
\end{align}
\item\label{obs:performance:iii} The additive performance of pattern~$p_z$ is at least the width~$b-a$ of strip~$[\fmin,\fmax]$. This performance is achieved if and only if~$x+y=a+b$, \ie, start point~$x$ and end point~$y$ are symmetric with respect to the center of interval~$[\fmin,\fmax]$.
\item\label{obs:performance:iv} The additive performance of pattern~$p_z$ is at most $\fmax-\fmin+\eps$ if and only if the pattern starts at~$x$ and ends at~$y\in[\yopt-\eps,\yopt+\eps]\cap[\fmin,\fmax]$ with~$\yopt:=a+b-x$.
\end{enumerate}
\end{observation}

\begin{proof}
Part~\eqref{obs:performance:i} follows from the fact that the maximum on the left hand side of~\eqref{eq:pattern-performance} is obtained at an index~$k$ for which~$p_z(k)$ is either maximal or minimal. The remaining statements follow immediately. 
\end{proof}

For any~$a'\in\BR$ and~$x\in[a',a'+\dmax]$, one can easily construct a pattern starting at~$x$ and living on a strip~$[a,b]\subseteq[a',a'+\dmax]$ of width at most~$\dmax$ by iteratively applying the following trivial observation.

\begin{observation}\label{obs:basic}
Let $d=u+v$ with $u,v\geq0$. If~$I$ is an interval of size at least~$d$ and~$x\in I$, then~$x+v\in I$ or~$x-u\in I$ (or both).	
\end{observation}

Choosing for example $a'=0$ and the start point~$x=\frac12\dmax$, Schrijver et al.~\cite{SchrijverSeymourWinkler-SIDMA} construct an unsplittable solution with additive performance at most~$\frac32\dmax$. In our terminology, their performance bound follows from Observation~\ref{obs:performance}~\eqref{obs:performance:iv}.

We notice that the described approach for constructing a pattern living on a strip of width at most~$\dmax$ can also be applied in backwards direction. In this way, for any~$a'\in\BR$ and~$y\in[a',a'+\dmax]$, one gets a pattern ending at~$y$ and living on a strip~$[a,b]\subseteq[a',a'+\dmax]$.

\section{Dealing with demands of medium size}
\label{sec:medium}

It is not difficult to observe that the additive performance of~$\frac32\dmax$ can be slightly improved in case there is a demand value of `medium size'.

\begin{lemma}\label{lem:medium-demands}
Let~$\delta\geq0$. If there is a demand~$i$ of value $d_i\in[\delta\dmax,(1-\delta)\dmax]$, then there exists a pattern with additive performance at most~$\bigl(\frac32-\frac\delta2\bigr)\dmax$. 
\end{lemma}

\begin{proof}
By reindexing the nodes of the ring accordingly, we may assume~$i=m$ without loss of generality. Deleting demand~$m$ (\eg, setting its demand value to~$d_m:=0$) yields a smaller instance. For this smaller instance we consider a pattern living on a sub-strip of~$[0,\dmax]$, ending at~$\frac12(\dmax+d_m)-v_m\in[0,\dmax]$, and starting at some~$x\in[0,\dmax]$. This pattern can be extended to a pattern for the original instance including demand~$m$ in two different ways: the first (choosing~$z_m=v_m$) ending at~$\frac12(\dmax+d_m)$ and the second (choosing~$z_m=-u_m$) ending at~$\frac12(\dmax-d_m)$. By construction, both patterns start at the same point~$x\in[0,\dmax]$ and live on a sub-strip of~$[0,\dmax]$; see Figure~\ref{fig:medium} for an illustration.
\begin{figure}[tb]
\centering
\begin{tikzpicture}[scale=1.4]
\draw[very thick,->] (-0.2,0) node[label=left:$0$] {} -- (5+0.3,0); 
\draw[thick,gray] (0,1.5) -- +(5.1,0); 
\draw[thick] (-0.1,1.5) node[label=left:$\dmax$] {} -- +(0.2,0); 
\draw[thick,gray,dashed] (0,1.2*0.75) -- +(4,0);
\draw[thick] (-0.1,1.2*0.75) node[label=left:$\frac12(\dmax+d_m)-v_m$] {} -- +(0.2,0); 

\foreach \i/\j in {0/0,1/1,2/\cdots,3/m-2,4/m-1,5/m}{
\draw [thick] (\i,-0.1) -- +(0,0.2); 
\node at (\i,-0.4) {$\j$};
}
\draw[very thick] (0,0) -- +(0,1.5+0.2); 
\draw[very thick] (5,0) -- +(0,1.5+0.2); 
\draw[very thick,gray] (0,0.5*0.75) node[point,label=left:$\color{black}x$] {} \foreach \z in {1.45,-1.2,-0.7,1.15} {-- ++(1,0.75*\z) node [point] {}}; 
\draw [very thick] (5,1.5*0.75) node [point,label=right:$\tfrac12(\dmax+d_m)$] {} -- node [sloped,above] {$+v_m$} (4,1.2*0.75) node [point] {} -- node [sloped,below] {$-u_m$} (5,0.5*0.75) node [point,label=right:$\tfrac12(\dmax-d_m)$] {};
\end{tikzpicture}		
\caption{Illustration of the two patterns constructed in the proof of Lemma~\ref{lem:medium-demands}.}
\label{fig:medium}
\end{figure}
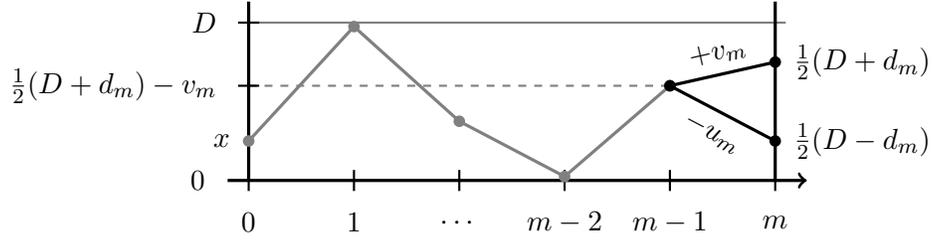

It follows from Observation~\ref{obs:performance}~\eqref{obs:performance:i} that the additive performance of the first pattern ending at~$\frac12(\dmax+d_m)$ is at most
\begin{align*}
\max\bigl\{2\dmax-x-\tfrac12(\dmax+d_m),x+\tfrac12(\dmax+d_m)\bigr\}\enspace.
\end{align*}
Since~$x\geq0$ and~$d_m\geq\delta\dmax$, the first of the two terms is at most~$\bigl(\frac32-\frac\delta2\bigr)\dmax$. Moreover, if~$x\leq\frac12\dmax$, then the same bound holds for the second term as~$d_m\leq(1-\delta)\dmax$. Finally, if~$x\geq\frac12\dmax$, the additive performance of the second pattern is at most~$\bigl(\frac32-\frac\delta2\bigr)\dmax$ by symmetric arguments.
\end{proof}


\section{Instances with small and big demands only}
\label{sec:small-big}

In this section we show how to obtain an improved additive performance for instances without medium size demand values. More precisely, for some~$0<\delta<\frac12$ to be determined later, we consider instances such that
\begin{align*}
d_i\in[0,\delta\dmax]\cup[(1-\delta)\dmax,\dmax]\qquad\text{for every demand~$i$.}
\end{align*}

The next definition will turn out to be useful in order to identify pairs of patterns that can easily be combined to a new pattern while keeping the width of the strip it is living on reasonably small.

\begin{definition}\label{def:closeness}
Let $\eps\geq0$. Two patterns~$p_{z'}$ and~$p_{z''}$ are said to be \emph{$\eps$-close} if $\lvert p_{z'}(k)-p_{z''}(k)\rvert\leq\eps$ for some~$k\in\{0,1,\dots,m\}$.
\end{definition}

The following lemma defines a useful crossover operation for two $\eps$-close patterns.

\begin{lemma}\label{lem:crossover}
Consider a fixed instance. Let~$p_{z'}$ be a pattern with start point~$x'$ living on strip~$[a',b']$, and $p_{z''}$ a pattern with end point~$y''$ living on strip~$[a'',b'']$. If the two patterns are $\eps$-close for some~$\eps>0$, then there is a pattern~$p_z$ living on a sub-strip of
\begin{align}\label{eq:lem:crossover}
\bigl[\min\{a',a''\}-\tfrac12\eps,\max\{b',b''\}+\tfrac12\eps\bigr]
\end{align}
with start point~$x$ and end point~$y$ such that~$x+y=x'+y''$.
\end{lemma}

\begin{proof}
Let~$k$ be chosen according to Definition~\ref{def:closeness} and let~$\eps':=p_{z'}(k)-p_{z''}(k)$ such that~$\lvert\eps'\rvert\leq\eps$. Define solution~$z$ by~$z_i:=z'_i$ for~$i=1,\dots,k$ and~$z_i:=z_i''$ for $i=k+1,\dots,m$. The corresponding pattern~$p_z$ with start point~$x:=x'-\frac12\eps'$ has end point~$y=y''+\frac12\eps'$ and lives on a sub-strip of~\eqref{eq:lem:crossover}; see Figure~\ref{fig:crossover} for an illustration.
\end{proof}
\begin{figure}[tb]
\centering
\begin{tikzpicture}[xscale=1.2,yscale=1.4]

\importfading{%
\begin{scope}
\includegraphics{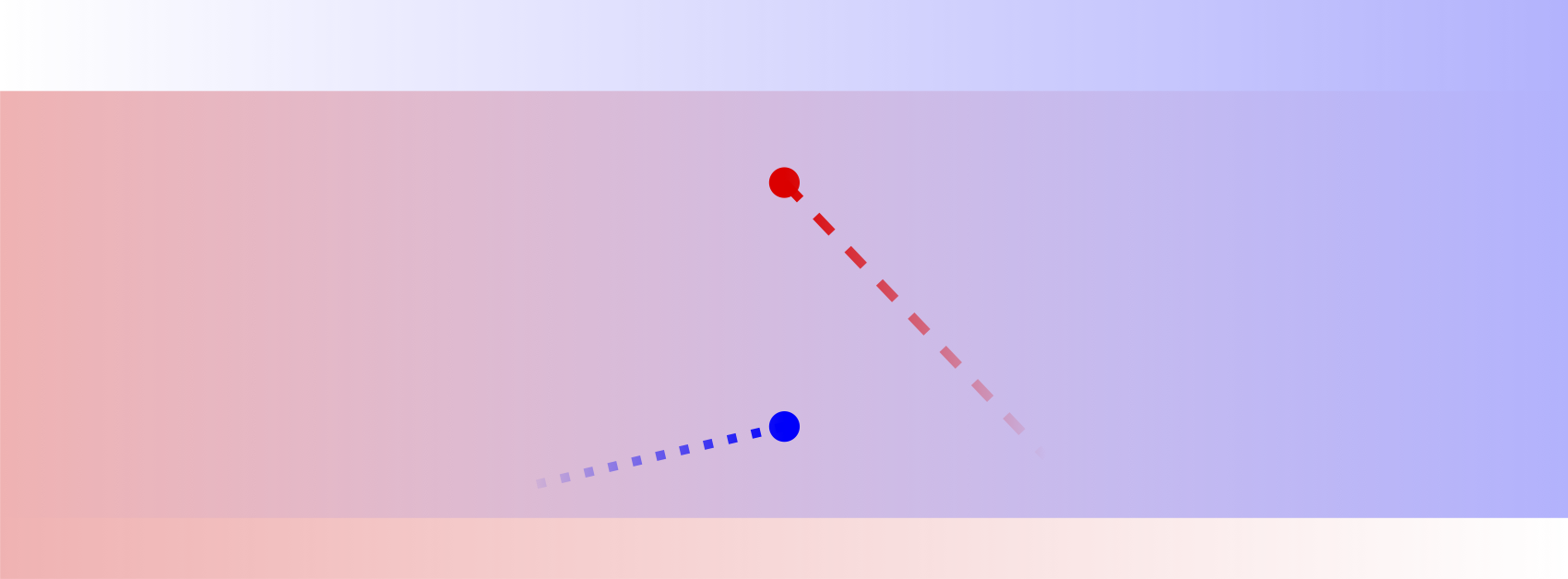}	
\end{scope}}{}
\begin{scope}[yshift=-0.2cm]
	
\importfading{}{%
\fill [fill=myred!60!white,path fading=east,fill opacity=0.5] (0,0.2) rectangle (6,1.8);
\fill [fill=myblue!60!white,path fading=west,fill opacity=0.5] (0,0.4) rectangle (6,2.1);}	
\draw[very thick,-] (-0.2,-0.4) -- +(6+0.4,0); 
\draw[thick] (-0.1,0.2) node[label=left:$a'$] {} -- +(0.2,0); 
\draw[thick] (-0.1,1.8) node[label=left:$b'$] {} -- +(0.2,0); 
\draw[thick] (5.9,0.4) -- +(0.2,0) node[label=right:$a''$] {}; 
\draw[thick] (5.9,2.1) -- +(0.2,0) node[label=right:$b''$] {}; 
\draw[thick,gray] (0,-0.2) -- +(6,0); 
\draw[thick] (-0.1,-0.2) node[label=left:$\min\{a'{,}a''\}-\frac12\eps'$] {} -- +(0.2,0); 
\draw[thick,gray] (0,2.5) -- +(6,0); 
\draw[thick] (-0.1,2.5) node[label=left:$\max\{b'{,}b''\}+\frac12\eps'$] {} -- +(0.2,0); 

\foreach \i/\j in {0/0,1/1,2/\dots,3/k,4/\dots,5/m-1,6/m} \draw [thick] (\i,-0.5) node[label=below:$\j$] {} -- +(0,0.2); 
\draw[very thick,-] (0,-0.5) -- +(0,3.1); 
\draw[very thick,-] (6,-0.5) -- +(0,3.1); 
\draw[very thick,myred,dashed] (0,0.9) node[point,label=left:$x'$] {} \foreach \z in {-0.7,1.4,-0.1} {-- ++(1,\z) node [point] {}}; 
\importfading{\node [point,myred] (rednode) at (3,1.5) {};}{\draw[very thick,myred,path fading=east,dashed] (3,1.5) node [point] (rednode) {} -- +(1,-0.9);}

\draw[very thick,myblue,dotted] (6,1.3) node[point,label=right:$y''$] {} \foreach \z in {+0.8,-0.6,-0.8} {-- ++(-1,\z) node [point] {}}; 
\importfading{\node [point,myblue] (bluenode) at (3,0.7) {};}{\draw[very thick,myblue,path fading=west,dotted] (3,0.7) node [point] (bluenode) {} -- +(-1,-0.2);}

\draw[very thick] (0,0.5) node[point,label=left:$x'-\frac12\eps'$] {} \foreach \z in {-0.7,1.4,-0.1} {-- ++(1,\z) node [point] {}}; 
\draw[very thick] (6,1.7) node[point,label=right:$y''+\frac12\eps'$] {} \foreach \z in {+0.8,-0.6,-0.8} {-- ++(-1,\z) node [point] {}}; 
\node [point] (blacknode) at (3,1.1) {};

\draw [<->] (blacknode) -- node[near end,left] {\footnotesize$\frac12\eps'$} (rednode);
\draw [<->] (blacknode) -- node[left] {\footnotesize$\frac12\eps'$} (bluenode);

\end{scope}
\end{tikzpicture}
\caption{Illustration of the crossover operation used in the proof of Lemma~\ref{lem:crossover}.}
\label{fig:crossover}
\end{figure}
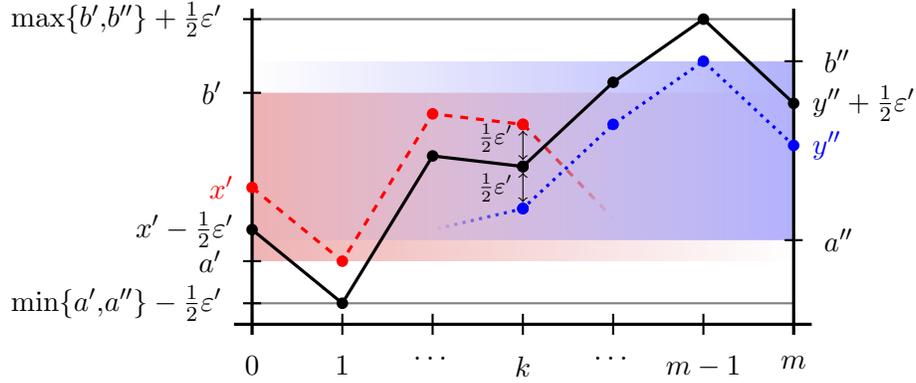

As discussed in Section~\ref{sec:preliminaries}, it follows from Observation~\ref{obs:basic} that for a given start point~$x=p_z(0)\in[0,\dmax]$ a pattern~$z$ living on a sub-strip of~$[0,\dmax]$ can be constructed by iteratively choosing~$z_k\in\{v_k,-u_k\}$ such that
\[
p_z(k)=p_z(k-1)+z_k\in[0,\dmax],\qquad\text{for~$k=1,\dots,m$.}
\]
Notice that the choice of~$z_k$ in this procedure is not always unique since both~$p_z(k-1)+v_k$ and~$p_z(k-1)-u_k$ may lie in strip~$[0,\dmax]$. In these situations a natural choice is to set~$z_k$ such that~$p_z(k)$ is as close as possible to the midpoint~$\dmax/2$ of strip~$[0,\dmax]$, that is, as far as possible away from the boundaries~$0$ and~$\dmax$. A pattern constructed according to this refined procedure (breaking remaining ties arbitrarily) is called a \emph{forward greedy pattern}. For technical reasons we also assume that the start point of a forward greedy pattern is far enough away from the borders of the strip~$[0,\dmax]$, that is~$p_z(0)\in\bigl[\frac14\delta\dmax,(1-\frac14\delta)\dmax\bigr]$.

Similarly, for a given end point~$y=p_z(m)\in\bigl[\frac14\delta\dmax,(1-\frac14\delta)\dmax\bigr]$, a \emph{backward greedy pattern}~$z$ living on a sub-strip of $[0,\dmax]$ is constructed by iteratively choosing~$z_k\in\{v_k,-u_k\}$ and setting~$p_z(k-1)=p_z(k)-z_k$ such that~$|p_z(k-1)-\frac12\dmax|$ is minimal, for~$k=m,\dots,1$. We say that~$z$ is a \emph{greedy pattern} if it is a forward or backward greedy pattern according to the definitions above. 

\begin{lemma}\label{lem:help}
Consider two greedy patterns~$p_z$ and $p_{z'}$. If there is a $k\in\{1,\dots,m\}$ such that
\begin{enumerate}[(i)]
\item $d_k\geq(1-\delta)\dmax$ and $z_k\neq z'_k$,
\item $p_z(k-1)\geq p_{z'}(k-1)$ and $p_z(k)\geq p_{z'}(k)$,
\end{enumerate}
then~$p_z$ and $p_{z'}$ are~$\frac12\delta\dmax$-close.
\end{lemma}

An illustration of a situation described in the lemma is given in Figure~\ref{fig:lemma_help}\,(a).
\begin{figure}[tb]
\centering
\begin{tikzpicture}[scale=1.5]
\draw [very thick,gray] (-0.25,1) node [black,left] {$\dmax$} -- +(2,0); 
\draw [very thick] (-0.25,-1) node [black,left] {$0$} -- +(2,0); 

\draw [thick] (0,-1.1) node[label=below:$k-1$] {} -- +(0,0.2);
\draw [thick] (1.5,-1.1) node[label=below:$k$] {} -- +(0,0.2);

\draw [very thick,myblue!40!white,dashed] (0,0.1) to +(1.5,-0.8) node [point] {};
\draw [very thick,myblue] (0,0.1) node [point] {} to node [above] {$p_{z}$} +(1.5,+0.5) node [point] {};
\draw [very thick,myred] (0,-0.1) node [point] {} to node [below] {$p_{z'}$} +(1.5,-0.8) node [point] {};

\draw [thick,myblue,decorate,decoration={brace,amplitude=5pt},xshift=0pt,yshift=0pt] (-0.1,0.1) -- +(0,0.5) node [black,midway,left,xshift=-3pt] {$v_k$};
\draw [thick,myred,decorate,decoration={brace,amplitude=5pt,mirror},xshift=0pt,yshift=0pt] (-0.1,-0.1) -- +(0,-0.8) node [black,midway,left,xshift=-3pt] {$u_k$};

\draw [thick,decorate,decoration={brace,amplitude=5pt},xshift=0pt,yshift=0pt] (1.6,0.57) -- +(0,-1.26) node [black,midway,right,xshift=3pt] {$d_k=v_k+u_k$};
\draw [thick,decorate,decoration={brace,amplitude=4pt},xshift=0pt,yshift=0pt] (1.6,0.97) -- +(0,-0.34) node [black,midway,right,xshift=3pt] {$\dmax-p_z(k)$};
\draw [thick,decorate,decoration={brace,amplitude=3pt},xshift=0pt,yshift=0pt] (1.6,-0.72) -- +(0,-0.26) node [black,midway,right,xshift=3pt] {$p_z(k)-d_k$};

\node at (1.2,-1.8) {(a)};

\begin{scope}[xshift=4.2cm]
\draw [very thick,gray] (-0.25,1) node [black,left] {$\dmax$} -- +(6,0); 
\draw [very thick] (-0.25,-1) node [black,left] {$0$} -- +(6,0); 

\draw [thick] (0,-1.1) node[label=below:$k-1$] {} -- +(0,0.2);
\draw [thick] (1.5,-1.1) node[label=below:$k$] {} -- +(0,0.2);
\draw [thick] (2.55,-1.1) node[label=below:$\cdots$] {} -- +(0,0.2);
\draw [thick] (3.6,-1.1) node[label=below:$q-1$] {} -- +(0,0.2);
\draw [thick] (5.1,-1.1) node[label=below:$q$] {} -- +(0,0.2);

\draw [very thick,myblue] (0,0.3) node [point] {} to node [above] {$p_{z}$} ++(1.5,+0.5) node [point] {} to ++(1.05,0.1) node [point] {} to ++(1.05,-0.1) node [point] {} to node [sloped,above,pos=0.4,yshift=-2pt] {$z_q=-u_q$} ++(1.5,-1.3) node [point] {};
\draw [very thick,myred] (0,-0.1) node [point] {} to node [below] {$p_{z'}$} ++(1.5,-0.8) node [point] {} to ++(1.05,0.1) node [point] {} to ++(1.05,-0.1) node [point] {} to node [sloped,above,pos=0.4,yshift=-2pt] {$z'_q=+v_q$} ++(1.5,0.7) node [point] {};

\draw [thick,decorate,decoration={brace,amplitude=4pt},xshift=0pt,yshift=0pt] (0.05,0.3) -- +(0,-0.4) node [black,midway,right,xshift=3pt] {$>\tfrac12\delta\dmax$};

\draw [thick,decorate,decoration={brace,amplitude=5pt},xshift=0pt,yshift=0pt] (1.55,0.78) -- +(0,-1.66) node [black,midway,right,xshift=3pt] {$>(1-\tfrac12\delta)\dmax$};

\draw [thick,decorate,decoration={brace,amplitude=5pt,mirror},xshift=0pt,yshift=0pt] (3.55,0.78) -- +(0,-1.66) node [black,midway,left,xshift=-3pt] {$<$};

\draw [thick,decorate,decoration={brace,amplitude=3pt},xshift=0pt,yshift=0pt] (5.15,-0.22) -- +(0,-0.26) node [black,midway,right,xshift=3pt] {$<\tfrac12\delta\dmax$};

\node at (2.7,-1.8) {(b)};
\end{scope}
\end{tikzpicture}
\caption{(a)~Two patterns~$p_z$ and~$p_{z'}$ featuring the properties stated in Lemma~\ref{lem:help}. Here we assume that~$z_k=v_k$ and~$z'_k=-u_k$. The other possibility ($z_k=-u_k$ and~$z'_k=v_k$) can be illustrated by mirroring both patterns left to right. (b)~Illustration of \emph{Case~III} in the proof of Lemma~\ref{lem:help}.}
\label{fig:lemma_help}
\end{figure}
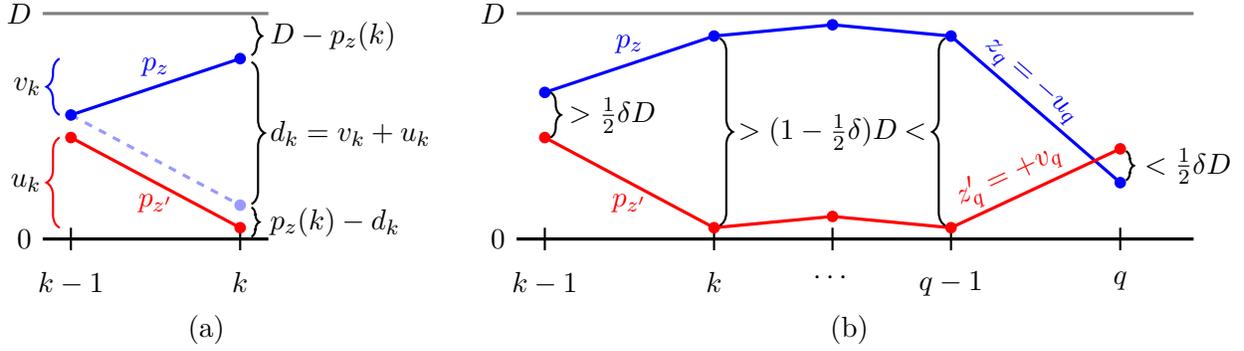
Even without the assumption that~$p_z$ and $p_{z'}$ are greedy patterns, it is not difficult to derive from Figure~\ref{fig:lemma_help}\,(a) a weaker version of this lemma that only guarantees $\delta\dmax$-closeness (remember that~$d_k=v_k+u_k\leq(1-\delta)\dmax$). In general, however, the stronger bound holds for the case of greedy patterns only.

\begin{proof}
We consider the scenario where~$z_k=v_k$ and thus~$z'_k=-u_k$ (see Figure~\ref{fig:lemma_help}\,(a)). The other scenario (\ie, $z_k=-u_k$ and~$z'_k=v_k$) can be transformed into our scenario by reading the patterns~$p_z$ and~$p_{z'}$ backwards from~$m$ to~$0$ (right to left). In particular, arguments symmetric to those given below can be used.

\emph{Case~I:} If $p_z$ is a \emph{forward} greedy pattern, then $p_z(k)-d_k\leq\dmax-p_z(k)$ (otherwise, during the construction of the forward greedy pattern~$p_z$, one would have chosen $z_k=-u_k$; see Figure~\ref{fig:lemma_help}\,(a)). This yields
\begin{multline*}
0
\leq p_z(k-1)-p_{z'}(k-1) 
= p_z(k)-v_k-\bigl(p_{z'}(k)+u_k\bigr)
\leq p_z(k)-(v_k+u_k)\\
= \tfrac12\bigl(p_z(k)-d_k\bigr)+\tfrac12\bigl(p_z(k)-d_k\bigr)
\leq\tfrac12\bigl(\dmax-p_z(k)\bigr)+\tfrac12\bigl(p_z(k)-d_k\bigr)
=\tfrac12(\dmax-d_k)
\leq\tfrac12\delta\dmax
\enspace.
\end{multline*}

\emph{Case~II:} Analogously, if~$p_{z'}$ is a \emph{forward} greedy pattern, then $\dmax-p_{z'}(k)-d_k\leq p_{z'}(k)$ (otherwise, during the construction of the forward greedy pattern~$p_{z'}$, one would have chosen $z'_k=v_k$). This yields
\begin{multline*}
0
\leq p_z(k-1)-p_{z'}(k-1) 
= p_z(k)-v_k-\bigl(p_{z'}(k)+u_k\bigr)
\leq\dmax-p_{z'}(k)-(v_k+u_k)\\
= \dmax-p_{z'}(k)-d_k
\leq\tfrac12 p_{z'}(k)+\tfrac12\bigl(\dmax-p_{z'}(k)-d_k\bigr)
=\tfrac12(\dmax-d_k)
\leq\tfrac12\delta\dmax
\enspace.
\end{multline*}

\emph{Case~III} (see Figure~\ref{fig:lemma_help}\,(b)): It remains to consider the case where both~$p_z$ and~$p_{z'}$ are \emph{backward} greedy patterns. We assume that $p_z(k-1)-p_{z'}(k-1)>\tfrac12\delta\dmax$; otherwise we are done. Then, 
\[
p_z(k)-p_{z'}(k)=d_k+p_z(k-1)-p_{z'}(k-1)>\bigl(1-\tfrac12\delta\bigr)\dmax\enspace.
\]
Since~$p_z(m),p_{z'}(m)\in\bigl[\frac14\delta\dmax,\bigl(1-\frac14\delta\bigr)\dmax\bigr]$ and thus~$p_z(m)-p_{z'}(m)\leq\bigl(1-\tfrac12\delta\bigr)\dmax$, there is a $q$ with~$k<q\leq m$ such that~$z_q=-u_q$ and~$z'_q=v_q$. Choosing the smallest such~$q$ we get
\[
p_z(q-1)-p_{z'}(q-1)\geq p_z(k)-p_{z'}(k)>\bigl(1-\tfrac12\delta\bigr)\dmax\enspace.
\]
As both~$p_z$ and $p_{z'}$ are backward greedy patterns, it must hold that~$p_z(q)\leq p_{z'}(q)$ since otherwise the choices $z_q=-u_q$ and~$z'_q=+v_q$ contradict the backward greedy paradigm. We can thus conclude that
\begin{multline*}
0\leq p_{z'}(q)-p_z(q)=p_{z'}(q-1)+v_q-\bigl(p_z(q-1)-u_q\bigr)\\
=d_q-\bigl(p_z(q-1)-p_{z'}(q-1)\bigr)<\dmax-\bigl(1-\tfrac12\delta\bigr)\dmax=\tfrac12\delta\dmax\enspace.
\end{multline*}
This concludes the proof.
\end{proof}

The next lemma identifies general situations where $\frac12\delta\dmax$-close patterns are guaranteed to exist.

\begin{lemma}\label{lem:cyclic-permutation}
Consider three greedy patterns~$p_{z^a},p_{z^b},p_{z^c}$, all three living on sub-strips of~$[0,\dmax]$. If the sorting of the patterns by their end points is not a cyclic permutation of the sorting by their start points, then (at least) two of the three patterns are $\frac12\delta\dmax$-close. 
\end{lemma}

\begin{proof}
Since the cyclic permutations form a subgroup of the symmetric group~$S_3$, there exists an index~$k\in\{1,\dots,m\}$ such that the sorting of patterns by their values~$p_{z^a}(k),p_{z^b}(k),p_{z^c}(k)$ is not a cyclic permutation of the sorting by their values~$p_{z^a}(k-1),p_{z^b}(k-1),p_{z^c}(k-1)$. This means that exactly two of the three patterns, say~$p_{z^a}$ and~$p_{z^b}$, must have exchanged their positions in the sorting when going from~$k-1$ to~$k$ while the third pattern~$p_{z^c}$ has kept its position in the sorting (here we use the fact that pairwise exchanges are the only non-cyclic permutations in the symmetric group~$S_3$). 

Since~$p_{z^a}$ and~$p_{z^b}$ have exchanged their position, $z^a_k$ and $z^b_k$ cannot be equal, that is, without loss of generality $z^a_k=v_k$ and~$z^b_k=-u_k$. Thus,~$p_{z^a}(k-1)\leq p_{z^b}(k-1)$ and~$p_{z^a}(k)\geq p_{z^b}(k)$. An illustration of the following case distinction is given in Figure~\ref{fig:cases}.

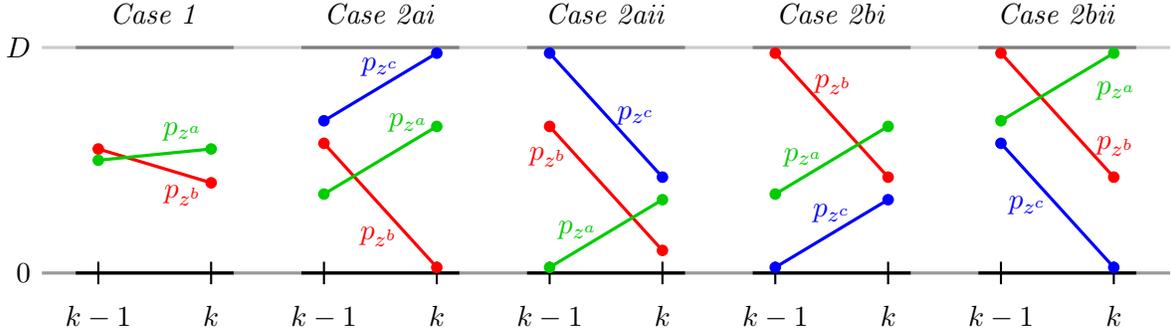
\begin{figure}[tb]
\centering
\begin{tikzpicture}[scale=1.5]
\draw [very thick,gray!40!white] (-0.5,1) node [black,left] {$\dmax$} -- +(10,0); 
\draw [very thick,gray!80!white] (-0.5,-1) node [black,left] {$0$} -- +(10,0); 

\foreach \i in {0,2,4,6,8}{
\draw [very thick] (\i-0.2,-1) -- +(1.4,0);
\draw [very thick,gray] (\i-0.2,1) -- +(1.4,0);
\draw [thick] (\i,-1.1) node[label=below:$k-1$] {} -- +(0,0.2);
\draw [thick] (\i+1,-1.1) node[label=below:$k$] {} -- +(0,0.2);
}

\begin{scope}[xshift=0cm]
\draw [very thick,myred] (0,0.1) node [point] {} to node [below,near end] {$p_{z^b}$} +(1,-0.3) node [point] {};
\draw [very thick,mygreen] (0,0) node [point] {} to node [above,near end] {$p_{z^a}$} +(1,0.1) node [point] {};
\node at (0.5,1.3) {\emph{Case~1}};
\end{scope}

\begin{scope}[xshift=2cm]
\draw [very thick,myred] (0,0.15) node [point] {} to node [left,near end] {$p_{z^b}$} +(1,-1.1) node [point] {};
\draw [very thick,mygreen] (0,-0.3) node [point] {} to node [above,near end] {$p_{z^a}$} +(1,0.6) node [point] {};
\draw [very thick,myblue] (0,0.35) node [point] {} to node [above] {$p_{z^c}$} +(1,0.6) node [point] {};	
\node at (0.5,1.3) {\emph{Case~2ai}};
\end{scope}

\begin{scope}[xshift=4cm]
\draw [very thick,myred] (0,0.3) node [point] {} to node [left,near start] {$p_{z^b}$} +(1,-1.1) node [point] {};
\draw [very thick,mygreen] (0,-0.95) node [point] {} to node [above,near start] {$p_{z^a}$} +(1,0.6) node [point] {};
\draw [very thick,myblue] (0,0.95) node [point] {} to node [right] {$p_{z^c}$} +(1,-1.1) node [point] {};	
\node at (0.5,1.3) {\emph{Case~2aii}};
\end{scope}

\begin{scope}[xshift=6cm]
\draw [very thick,myred] (0,0.95) node [point] {} to node [right,near start] {$p_{z^b}$} +(1,-1.1) node [point] {};
\draw [very thick,mygreen] (0,-0.3) node [point] {} to node [above,near start] {$p_{z^a}$} +(1,0.6) node [point] {};
\draw [very thick,myblue] (0,-0.95) node [point] {} to node [above] {$p_{z^c}$} +(1,0.6) node [point] {};	
\node at (0.5,1.3) {\emph{Case~2bi}};
\end{scope}

\begin{scope}[xshift=8cm]
\draw [very thick,myred] (0,0.95) node [point] {} to node [right,near end] {$p_{z^b}$} +(1,-1.1) node [point] {};
\draw [very thick,mygreen] (0,0.35) node [point] {} to node [below right,near end] {$p_{z^a}$} +(1,0.6) node [point] {};
\draw [very thick,myblue] (0,0.15) node [point] {} to node [left] {$p_{z^c}$} +(1,-1.1) node [point] {};	
\node at (0.5,1.3) {\emph{Case~2bii}};
\end{scope}

\end{tikzpicture}	
\caption{Illustration of different cases discussed in the proof of Lemma~\ref{lem:cyclic-permutation}.}
\label{fig:cases}
\end{figure}

\emph{Case~1:} If~$d_k=u_k+v_k\leq\delta\dmax$, we show that~$p_{z^a}$ and~$p_{z^b}$ are~$\frac12\delta\dmax$-close:
\begin{multline*}
|p_{z^a}(k)-p_{z^b}(k)|+|p_{z^b}(k-1)-p_{z^a}(k-1)|\\=\bigl(p_{z^a}(k)-p_{z^a}(k-1)\bigr)-\bigl(p_{z^b}(k)-p_{z^b}(k-1)\bigr)=v_k+u_k=d_k\leq\delta\dmax\enspace.
\end{multline*}
This inequality implies that at position~$k-1$ or~$k$ the two patterns are at most~$\frac12\delta\dmax$ apart.

\emph{Case~2:} It remains to consider the case where~$d_k\geq(1-\delta)\dmax$. With the help of Lemma~\ref{lem:help} we argue that~$p_{z^c}$ is $\frac12\delta\dmax$-close to at least one of the other two patterns in this case. First notice that the unchanged position of pattern~$p_{z^c}$ in the sorting cannot be the second position between the other two patterns as~$z^c_k\in\{v_k,-u_k\}$ and thus its position relative to one of the other patterns does not change when going from~$k-1$ to~$k$. Therefore,~$p_{z^c}(k-1)$ and~$p_{z^c}(k)$ are either both maximal or both minimal. 

\emph{Case~2a:} We first consider the case that $p_{z^c}(k-1)$ and~$p_{z^c}(k)$ are both maximal. If~$z^c_k=v_k$ (\emph{Case~2ai} in Figure~\ref{fig:cases}), then~$p_{z^b}$ and~$p_{z^c}$ are~$\frac12\delta\dmax$-close by Lemma~\ref{lem:help}. If, on the other hand,~$z^c_k=-u_k$ (\emph{Case~2aii} in Figure~\ref{fig:cases}), then~$p_{z^a}$ and~$p_{z^c}$ are~$\frac12\delta\dmax$-close by Lemma~\ref{lem:help}. 

\emph{Case~2b:} The same arguments can be applied to the symmetric case where~$p_{z^c}(k-1)$ and~$p_{z^c}(k)$ are both minimal; see Cases~2bi and~2bii in Figure~\ref{fig:cases}.
\end{proof}

The next lemma is the key to proving the existence of unsplittable solutions with improved additive performance (Theorem~\ref{thm:7/5}). From now on we set~$\delta:=\frac27$, which can be seen to be optimal for our kind of analysis.

\begin{lemma}\label{lem:large-demands}
For an instance with~$d_i\in\bigl[0,\frac27\dmax\bigr]\cup\bigl[\frac57\dmax,\dmax\bigr]$ for every demand~$i$, there is a pattern 
with additive performance at most~$\frac{19}{14}\dmax$.
\end{lemma}

\begin{proof}
We start by considering a forward greedy pattern~$p_{z^a}$ with start point~$x^a:=\frac5{14}\dmax$, living on a sub-strip of~$[0,\dmax]$. If its end point~$y^a$ is at least~$\frac27\dmax$, then its additive performance is at most~$\frac{19}{14}\dmax$ by Observation~\ref{obs:performance}~\eqref{obs:performance:iv} and we are done. We therefore assume that~$y^a\in\bigl[0,\frac27\dmax\bigr]$.

Our second pattern~$p_{z^b}$ is a backward greedy pattern that ends at $y^b:=\frac37\dmax$. Similar to above, we can assume that its start point~$x^b\in\bigl[0,\frac3{14}\dmax\bigr]$. If the two greedy patterns introduced so far are~$\frac17\dmax$-close, then, by Lemma~\ref{lem:crossover}, there is a pattern living on a sub-strip of~$\bigl[-\frac1{14}\dmax,\frac{15}{14}\dmax\bigr]$ with start point~$x$ and end point~$y$ such that~$x+y=x^a+y^b=\frac{11}{14}\dmax$. By Observation~\ref{obs:performance}~\eqref{obs:performance:i}, the additive performance of this pattern is at most
\begin{align*}
\max\bigl\{2\cdot\tfrac{15}{14}\dmax-\tfrac{11}{14}\dmax,\tfrac{11}{14}\dmax+2\cdot\tfrac1{14}\dmax\bigr\}=\tfrac{19}{14}\dmax
\end{align*}
and we are done. We therefore assume in the following that~$p_{z^a}$ and~$p_{z^b}$ are not~$\frac17\dmax$-close.

We now consider another forward greedy pattern~$p_{z^c}$ with start point~$x^c:=\frac{11}{14}\dmax$. If its end point~$y^c$ is at most~$\frac47\dmax$, then by Observation~\ref{obs:performance}~\eqref{obs:performance:iv} its additive performance is at most~$\frac{19}{14}\dmax$ and we are done. We can therefore assume that~$y^c\in\bigl[\frac47\dmax,\dmax\bigr]$.

Notice that the sorting of the three patterns by their end points (\ie, $a,b,c$) is not a cyclic permutation of the sorting by their start points (\ie, $b,a,c$). Therefore, by Lemma~\ref{lem:cyclic-permutation} with~$\delta:=\frac27$, at least two of the three patterns are $\frac17\dmax$-close. By our assumption above,~$p_{z^a}$ and~$p_{z^b}$ are not~$\frac17\dmax$-close. It therefore remains to consider two cases.

If~$p_{z^c}$ and~$p_{z^a}$ are~$\frac17\dmax$-close, then, by Lemma~\ref{lem:crossover}, there is a pattern living on a sub-strip of~$\bigl[-\frac1{14}\dmax,\frac{15}{14}\dmax\bigr]$ with start point~$x$ and end point~$y$ such that~$x+y=x^c+y^a\in\bigl[\frac{11}{14}\dmax,\frac{15}{14}\dmax\bigr]$. By Observation~\ref{obs:performance}~\eqref{obs:performance:i}, the additive performance of this pattern is at most
\begin{align*}
\max\bigl\{2\cdot\tfrac{15}{14}\dmax-\tfrac{11}{14}\dmax,\tfrac{15}{14}\dmax+2\cdot\tfrac1{14}\dmax\bigr\}=\tfrac{19}{14}\dmax
\end{align*}
and we are done. Finally, if~$p_{z^c}$ and~$p_{z^b}$ are~$\frac17\dmax$-close, then, by Lemma~\ref{lem:crossover}, there is a pattern living on a sub-strip of~$\bigl[-\frac1{14}\dmax,\frac{15}{14}\dmax\bigr]$ with start point~$x$ and end point~$y$ such that~$x+y=x^c+y^b=\frac{17}{14}\dmax$. By Observation~\ref{obs:performance}~\eqref{obs:performance:i}, the additive performance of this pattern is at most
\begin{align*}
\max\bigl\{2\cdot\tfrac{15}{14}\dmax-\tfrac{17}{14}\dmax,\tfrac{17}{14}\dmax+2\cdot\tfrac1{14}\dmax\bigr\}=\tfrac{19}{14}\dmax\enspace.
\end{align*}
This concludes the proof.
\end{proof}

We can finally prove our main result.

\begin{proof}[Proof of Theorem~\ref{thm:7/5}]
If the considered instance has a demand~$i$ of value~$d_i\in\bigl[\frac27\dmax,\frac57\dmax\bigr]$, the result follows by Lemma~\ref{lem:medium-demands} with~$\delta:=\frac27\dmax$. Otherwise, the result follows from Lemma~\ref{lem:large-demands}.
\end{proof}

Notice that the algorithm that is implicitly given in the proof of Theorem~\ref{thm:7/5} can be implemented to run in linear time.

\section{Deriving an improved lower bound}
\label{sec:lower-bound}

In this section we present a counterexample to Schrijver et al.'s conjecture. While this counterexample turns out to be relatively simple and easy to verify, we would like to emphasize that finding it was a cumbersome undertaking for both the author and his computer. Schrijver et al.\ write in their paper~\cite{SchrijverSeymourWinkler-SIDMA} ``We have never managed to produce a random example with $L>L^*+D$\dots'', and we can certainly confirm this experience. Among hundreds of millions of similar instances of the same size, all featuring a very special, promising structure, the one presented below turned out to be the only counterexample. 

In the following we first give some more details on how we tracked down this shy beast. In fact, our primary objective was to find a split routing that cannot be turned into an unsplittable routing without increasing the load on some edge by~$(1+\eps)\dmax$ for~$\eps>0$ as large as possible. Only later we realized that our insights can also be used to disprove Schrijver et al.'s conjecture.

Schrijver, Seymour, and Winkler~\cite{SchrijverSeymourWinkler-SIDMA} exhibit an instance of the \RLP{} on a ring of $n=26$~nodes and~$m=13$ nonzero demands together with a split routing that cannot be turned into an unsplittable routing without increasing the load on some edge by at least~$\frac{101}{100}\dmax$. An illustration of a slightly simplified version of this instance with~$n=24$ nodes and~$m=12$ demands can be found in Figure~\ref{fig:schrijver-instance}.
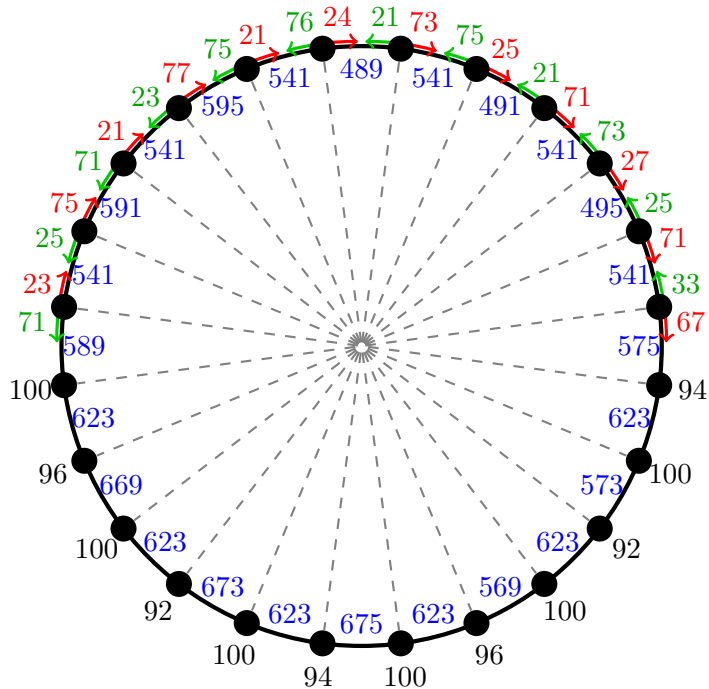
\begin{figure}[tb]
\centering
\begin{tikzpicture}
\def\radius{3.99}
\pgfmathparse{\radius+0.45}\let\blacklabelradius\pgfmathresult
\pgfmathparse{\radius+0.4}\let\graylabelradius\pgfmathresult
\pgfmathparse{\radius+0.07}\let\grayarcradius\pgfmathresult
\pgfmathparse{\radius-0.3}\let\flowradius\pgfmathresult
\def\grayangle{9}
	
\drawring{71/23,25/75,71/21,23/77,75/21,76/24,21/73,75/25,21/71,73/27,25/71,33/67}{}
\end{tikzpicture}
\caption{A slightly simplified version of the Schrijver et al.\ instance with~$\dmax=100$. The given split routing cannot be turned into an unsplittable routing without increasing the load on some edge by at least~$101$.}
\label{fig:schrijver-instance}
\end{figure}	
The maximum demand value is~$\dmax=100$, and it is a nontrivial task to verify that in any unsplittable routing there is an edge whose load exceeds the load of the split routing given in Figure~\ref{fig:schrijver-instance} by at least~$101$. With the help of a computer and brute force, however, one can either check all~$2^{12}$ unsplittable routings or, somewhat less brute, use dynamic programming to find an unsplittable routing with minimum possible increase of any edge load which can be done in time polynomial in~$n$ and~$\dmax$.\footnote{The idea is to enumerate over all integral patterns living on a sub-strip of~$[0,\dmax]$ using dynamic programming.}

In order to increase this lower bound of~$\frac{101}{100}$, it is crucial to gain some understanding of the idea behind Schrijver et al.'s example: If there was an unsplittable routing~$z$ that increases the load on every edge by at most~$\dmax=100$, it must have a corresponding pattern~$p_z$ that lives on a sub-strip of~$[0,100]$, has some start point~$x$ and `symmetric' end point~$100-x$. Schrijver et al.'s instance, however, is constructed such that all values~$u_i$ and~$v_i$ are positive integers, $u_i+v_i$ is even, and the difference of the start and the end point of any pattern is odd, as~$\sum_{i=1}^{12}u_i$ is odd. Since for integral~$x$ the difference of~$x$ and~$100-x$ is always even, the pattern would thus have to use half-integral points only and, as a consequence, lives on a sub-strip of~$[0.5,99.5]$ whose width is at most~$99$. This restriction of the width of suitable patterns together with the symmetry requirement on start and end points is the key property of Schrijver et al.'s instance. In this context it is important to notice that every other demand value is equal to~$100$ which essentially forces patterns to grow wide.

Notice that the example on~$n=16$ nodes depicted in Figure~\ref{fig:first-example}, with~$m=8$ demands and~$\dmax=10$ features the same three properties as Schrijver et al.'s example: (i)~all values~$u_i$ and~$v_i$ are positive integers with $u_i+v_i$ even; (ii)~start and end point of any pattern differ by an odd number as~$\sum_{i=1}^8u_i$ is odd; (iii)~every other demand has value~$\dmax=10$. For~$m=8$ and~$\dmax=10$, there exist more than one billion 
such sequences~$(u_1,v_1),\dots,(u_8,v_8)$. Taking certain symmetries into account, this number can be reduced to several hundreds of millions. By brute force enumeration, we found exactly two split routings that cannot be turned into unsplittable routings without increasing the load on some edge by at least~$11=\dmax+1$: the one depicted in Figure~\ref{fig:first-example} and the one depicted in Figure~\ref{fig:11/10}.
\begin{figure}[t]
\centering
\begin{tikzpicture}[scale=1.0]
\def\radius{3.45}
\pgfmathparse{\radius+0.45}\let\blacklabelradius\pgfmathresult
\pgfmathparse{\radius+0.4}\let\graylabelradius\pgfmathresult
\pgfmathparse{\radius+0.07}\let\grayarcradius\pgfmathresult
\pgfmathparse{\radius-0.3}\let\flowradius\pgfmathresult
\def\grayangle{9}

\drawring{2/2, 3/7, 7/1, 3/7, 2/2, 4/6, 4/4, 6/4}{}	
\end{tikzpicture}
\caption{A split routing for $8$ demands between opposite nodes together with the resulting edge loads. Each of the~$256=2^8$ unsplittable routings increases the load on some edge by at least~$11=\dmax+1$.}
\label{fig:11/10}
\end{figure}
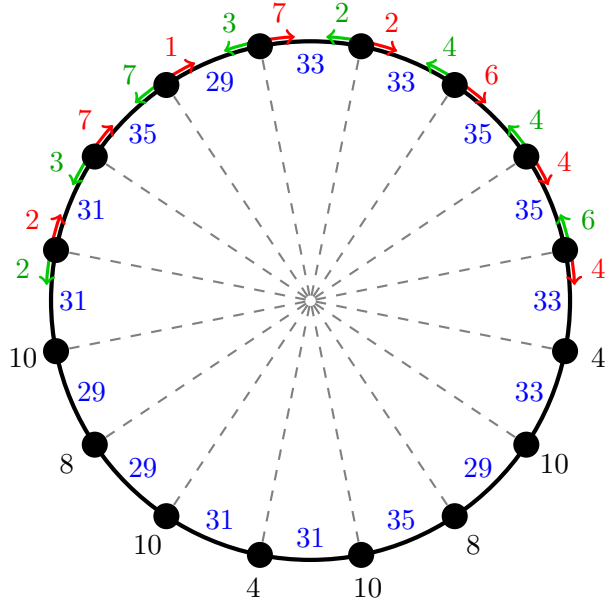

\begin{observation}\label{obs:11/10}
There exist instances with a given split routing such that in any unsplittable solution the load on some edge is increased by at least $\frac{11}{10}\dmax$ (cf.~Figures~\ref{fig:first-example} and~\ref{fig:11/10}).
\end{observation}

This observation does, however, not yet imply the existence of a counterexample to Schrijver et al.'s conjecture. Firstly, the split routings in Figures~\ref{fig:first-example} and~\ref{fig:11/10} are not optimum (notice that an optimum split routing for pairwise crossing demands splits every demand into two equal pieces). And, secondly, it is not true that any unsplittable routing increases the \emph{maximum} edge load by more than~$\dmax$. 

For the split routings in Figures~\ref{fig:first-example} and~\ref{fig:11/10} the first problem can be tackled by introducing for every edge an additional demand between its two end nodes. The demand values are chosen such that routing these demands unsplittably along their corresponding edges equalizes all edge loads; see Figure~\ref{fig:non-disproving-instance} for the necessary modification of the instance depicted in Figure~\ref{fig:first-example}. 
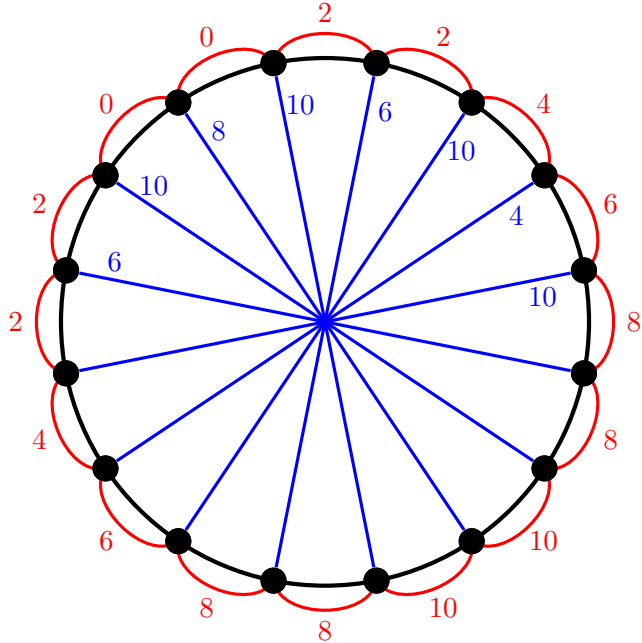
\begin{figure}[t]
\centering
\def\radius{3.51}
\pgfmathparse{\radius-0.6}\let\bluedemandradius\pgfmathresult
\pgfmathparse{\radius+0.6}\let\reddemandradius\pgfmathresult
\begin{tikzpicture}[xscale=-1]
\draw [ultra thick] (0,0) circle (\radius cm);
\foreach \i/\j/\k in {1/10/6,2/4/4,3/10/2,4/6/2,5/10/0,6/8/0,7/10/2,8/6/2}{
\draw [very thick,myred] (191.25-\i*22.5:\radius cm) to [bend left=80] (191.25-\i*22.5-22.5:\radius cm) node [node,black] {};
\node [myred] at (180-\i*22.5:\reddemandradius cm) {$\k$};
\draw [very thick,myred] (11.25-\i*22.5:\radius cm) to [bend left=80] (11.25-\i*22.5-22.5:\radius cm) node [node,black] {};
\pgfmathparse{int(10-\k)}\let\k\pgfmathresult
\node [myred] at (-\i*22.5:\reddemandradius cm) {$\k$};

\node [node] (n\i) at (191.25-\i*22.5:\radius cm) {};
\node [node] (nn\i) at (11.25-\i*22.5:\radius cm) {};
\draw [very thick,myblue] (n\i) -- (nn\i);
\node [myblue] at (196-\i*22.5:\bluedemandradius cm) {$\j$};
}
\end{tikzpicture}	
\caption{Extension of the instance depicted in Figure~\ref{fig:first-example}. By routing the additional `short' demands unsplittably along their corresponding edges and the remaining demands between opposite nodes as suggested in Figure~\ref{fig:first-example} we get an optimum split routing where all edges have load~$37$.}
\label{fig:non-disproving-instance}
\end{figure}
We argue that the split routing with uniform edge loads minimizes the maximum edge load. This follows from the fact that all demands are routed along shortest paths only such that the average edge load is minimum. Since the minimum and the maximum edge loads in Figure~\ref{fig:first-example} (and also in Figure~\ref{fig:11/10}) differ by no more than~$\dmax$, the additional demand values can be chosen to be at most~$\dmax$. (Notice that this is not true for the split routing depicted in Figure~\ref{fig:schrijver-instance}.) 

This careful modification of the considered instances almost solves the second problem as well. As a consequence of Observation~\ref{obs:11/10}, any unsplittable routing increases the maximum edge load by at least~$\frac{11}{10}\dmax$ as long as it sticks to routing the additional demands along their corresponding edges. The latter restriction, however, turns out to be crucial. The instance depicted in Figure~\ref{fig:non-disproving-instance} (based on Figure~\ref{fig:first-example}) has an unsplittable routing whose maximum edge load exceeds the maximum edge load of the optimum split routing by only~$\dmax=10$.

On the other hand, the extension of the instance depicted in Figure~\ref{fig:11/10} does lead to a counterexample, a slightly simplified version of which is depicted in Figure~\ref{fig:simpler-disproving-instance}. 
%
%
%
\begin{figure}[t]
\centering
\def\radius{3.44}
\pgfmathparse{\radius-0.6}\let\bluedemandradius\pgfmathresult
\pgfmathparse{\radius+0.7}\let\reddemandradius\pgfmathresult
\pgfmathparse{\radius+0.8}\let\newreddemandradius\pgfmathresult
\begin{tikzpicture}[xscale=-1,scale=0.95]
\draw [ultra thick] (0,0) circle (\radius cm);
\foreach \i/\j in {1/10,2/8,3/10,4/4,5/10,6/8,7/10,8/4}{
\node [node] (n\i) at (191.25-\i*22.5:\radius cm) {};
\node [node] (nn\i) at (11.25-\i*22.5:\radius cm) {};
\draw [very thick,myblue] (n\i) -- (nn\i);
\node [mybluelabel] at (196-\i*22.5:\bluedemandradius cm) {$\j$};
}
\foreach \i/\k in {5/10,6/4}{
\draw [very thick,myred] (191.25-\i*22.5:\radius cm) node [node,black] {} to [bend left=80] (191.25-\i*22.5-22.5:\radius cm) node [node,black] {};
\node [myred] at (180-\i*22.5:\reddemandradius cm) {$\k$};
\draw [very thick,myred] (11.25-\i*22.5:\radius cm) node [node,black] {} to [bend left=80] (11.25-\i*22.5-22.5:\radius cm) node [node,black] {};
\pgfmathparse{int(14-\k)}\let\k\pgfmathresult
\node [myred] at (-\i*22.5:\reddemandradius cm) {$\k$};
}
\foreach \i/\k in {1/4,3/6,7/8}{
\draw [very thick,myred] (191.25-\i*22.5:\radius cm) node [node,black] {} to [bend left=80] (191.25-\i*22.5-45:\radius cm) node [node,black] {};
\node [myred] at (180-\i*22.5-11.25:\newreddemandradius cm) {$\k$};
\draw [very thick,myred] (11.25-\i*22.5:\radius cm) node [node,black] {} to [bend left=80] (11.25-\i*22.5-45:\radius cm) node [node,black] {};
\pgfmathparse{int(14-\k)}\let\k\pgfmathresult
\node [myred] at (-\i*22.5-11.25:\newreddemandradius cm) {$\k$};
}
\end{tikzpicture}	
\caption{An instance of the ring loading problem on~$n=16$ nodes with~$18$ demands, $8$ of which belong to opposite nodes,~$4$ to neighboring nodes, and~$6$ to nodes of distance~$2$. The maximum demand value is~$\dmax=10$ and it can be checked that an optimum split routing has load~$L^*=39$ while an optimum unsplittable routing has load~$L=L^*+\dmax+1=50$.}
\label{fig:simpler-disproving-instance}
\end{figure}
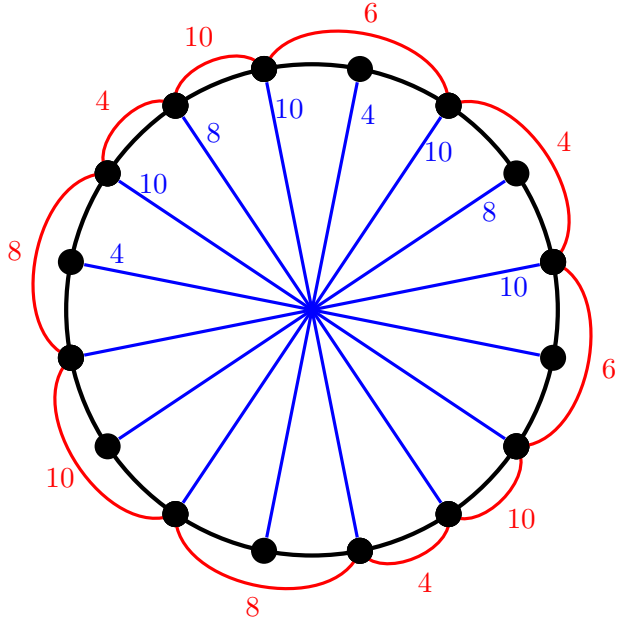
Similar to above, it is easy to see that combining the split routing in Figure~\ref{fig:11/10} with the straightforward `short' routing of the additional demands yields an optimum split routing where all edge loads are equal to~$L^*=39$.
%
Finally, by enumerating all~$2^{18}=262.144$ unsplittable routings, one can verify that the optimal unsplittable load is~$L=50$ such that~$L=L^*+\dmax+1$, thus disproving Schrijver et al.'s conjecture that~$L\leq L^*+\dmax$.

\section{Conclusion}
\label{sec:conclusion}

We have shown that the worst-case additive gap between optimum split and unsplittable routings for the \RLP{} is at most~$+\frac{19}{14}\dmax$ and at least~$+\frac{11}{10}\dmax$. That is, $L\leq L^*+\frac{19}{14}\dmax$ for all instances, and there exists an instance with~$L\geq L^*+\frac{11}{10}\dmax$. Our extensive yet unsuccessful search for instances yielding a larger lower bound gives us serious doubts as to whether there exist any.

\begin{conjecture}
$L\leq L^*+\frac{11}{10}\dmax$ for all instances of the \RLP{}.
\end{conjecture}

We have serious doubts as to whether the algorithmic techniques and analytic tools discussed in this paper are powerful enough to close the remaining gap. But we hope that, 16 years after the publication of Schrijver et al.'s landmark paper~\cite{SchrijverSeymourWinkler-SIDMA}, the presented progress will stimulate further research and new ideas on this fine and challenging problem.

In view of Shepherd's remark (cf.~\cite{Shepherd-survey2009}) on the `tempting similarity' of Schrijver et al.'s conjecture and results on single-source unsplittable flows~\cite{DGG99} and, in particular, Goemans' related conjecture~\cite{BellairsOpenProblems}, it is natural to ask for possible consequences of the insights presented in this paper. We do not have a meaningful answer to this question, nor do we currently see any reason to doubt the validity of Goemans' conjecture.

\paragraph{Acknowledgements.} The author would like to thank Bruce Shepherd for giving the impetus to this work as well as for numerous inspiring conversations and valuable comments.

\bibliographystyle{plain}
\bibliography{../BibTeX/mybib.bib}

\end{document}